%% file: arXiv_version_finitary_bounded_degree.tex
\documentclass[11pt]{article}

\usepackage{alltheorems}
\usepackage{hyperref}

\usepackage{authblk}

\usepackage[dvipsnames]{xcolor}
\usepackage{amsfonts}
\usepackage{amssymb}

\usepackage{amsmath}
\usepackage{constants}
\usepackage{amsthm}
\usepackage{xspace}
\usepackage{paralist}
\usepackage{tikz}
\usetikzlibrary{shapes,calc,positioning,arrows,hobby}
\makeatletter
\newcommand{\newreptheorem}[2]{%
\newenvironment{rep#1}[1]{%
 \def\rep@title{#2 \ref{##1}}%
 \begin{rep@theorem}}%
 {\end{rep@theorem}}}
\makeatother

\usepackage{ifpdf}

\ifpdf
 
\fi

\usepackage{tikz}
\usetikzlibrary{shapes,calc,positioning,arrows,}

\theoremstyle{plain}
\input{macros}

\title{Constant time testability of first-order logic with modulo counting on finitary graphs}

\author{Isolde Adler\footnote{University of Bamberg, email: isolde.adler@uni-bamberg.de, \\ orcid: 0000-0002-9667-9841} \and Jenny Stimpson\footnote{University of Bamberg, email: j.stimpson@uni-bamberg.de, \\ orcid: 0009-0008-4032-2123}}

\date{}

\begin{document}
	

\begin{titlepage}
		
\maketitle
		
\thispagestyle{empty}
\begin{abstract}
	This paper studies algorithmic meta theorems for property testing with \emph{constant
	running time} in the bounded degree model.
	In (Adler, Harwath 2018) it was shown that
	on graph classes $\mathcal C^{w}_d$ consisting of all graphs with both 
	degree at most $d$ and treewidth at most $w$, every problem expressible in
	\emph{monadic second-order logic with counting} ($\cmso$) is testable with \emph{polylogarithmic} 
	running time (where $d,w\in \mathbb N$ are fixed).
	It was left open whether this can be improved to \emph{constant} running time.
	
	In this paper we give a positive answer for testing $\cmso$ on classes $\mathcal C^{c}_d$, where
	$d$ bounds the degree and $c$ bounds the component size.
	Our main result shows constant time testability of \emph{first-order logic with modulo counting} ($\fomod$)
	on $\mathcal C^{c}_d$. For our proof we tailor
	Hanf normal form of 
	$\fomod$ to our setting, and we exhibit a number-theoretic `patchability' condition that
	allows to infer global information on the input graph from a local sample of constant size. We
	believe that our `patchability' might be of independent interest.
	The step from $\fomod$ to $\cmso$ then follows from a result by (Eickmeyer, Elberfeld, Harwath, 2017) 
	on the expressive power of order invariant monadic second-order logic on classes of bounded treedepth.
	
	\bigskip
	\textbf{Keywords:} Property testing, constant time algorithms, bounded degree model, algorithmic meta-theorems, first-order logic with modulo counting, monadic second-order logic.

\end{abstract}

\end{titlepage}

\section{Introduction}\label{sec:intro}

The past decades have seen a tremendous increase in the volume of stored data, posing a huge challenge to computational power. Moreover, the data is often stored remotely, making access costly. In this context, even just reading the whole input once may be infeasible. \emph{Property Testing} addresses this by the design of algorithms
that derive global information on the structure of the input by only exploring a small number of
local parts of it. These algorithms are randomised and allow for a small error. 
A property testing algorithm is required to 
distinguish inputs that satisfy a given property from inputs that
are $\epsilon$-far from satisfying the property with high probability correctly. Intuitively,
an input is called $\epsilon$-far from satisfying a property if more than an $\epsilon$-fraction of the input 
needs to be modified in
order to make it satisfy the property. The exact definition of the distance measure depends
on the model. 

Our work is based on the \emph{bounded degree graph model} of property testing introduced
by Goldreich and Ron~\cite{GoldreichR02}. 
This model considers graphs whose degree (\ie the number of edges incident to any vertex) 
is bounded from above by a constant $d\in \mathbb N$.
Formally, a \emph{property} $P$ is a class of graphs that is closed 
under isomorphism.
Let $\epsilon\in (0,1)$. Two $d$-bounded degree graphs $G$ and $H$, both on $n$ vertices, 
are called \emph{$\epsilon$-close}, if they can 
be made isomorphic by modifying (adding or deleting) at most $\epsilon dn$ edges. 
If $G$ and $H$ are not $\epsilon$-close they are called \emph{$\epsilon$-far}. For a property $P$ of $d$-bounded degree graphs, a
graph $G$ is \emph{$\epsilon$-close to $P$}, if $G$ is $\epsilon$-close to some member of $P$. 
Otherwise $G$ is \emph{$\epsilon$-far from $P$}.
The algorithms have \emph{oracle access} to an input graph $G$ of maximum degree $d$. Here we assume that 
the vertex set of $G$ is $\{1,2,\ldots,n\}$. More precisely, an algorithm gets the number $n$ as input, and
it can access $G$ via \emph{neighbour queries} to an oracle, \ie for any specified vertex $v$ and index $i\in \{1,2,\ldots,d\}$, the oracle returns the $i$-th neighbour of $v$ if it exists, and a special symbol $\bot$ otherwise, in constant time. 

A property $P$ is \emph{testable with query complexity $q(n)$} in the bounded degree model, if for every $\varepsilon\in (0,1)$ 
and every input size $n$ there is a probabilistic algorithm $A_{\epsilon,n}$ (called an \emph{$\varepsilon$-tester}),
that accepts $G$ with probability at least $2/3$, if $G\in P$, and rejects $G$ with probability at least $2/3$ if $G$ is $\epsilon$-far from $P$, 
while using at most $q(n)$ oracle queries. If $q$ is independent of $n$, then
$P$ is \emph{constant-query testable}. (Here $q$ may depend on $\epsilon$, $d$ and $P$.)
Until now, despite some effort, a \emph{characterisation of the properties that are constant-query testable} in the bounded degree model is still elusive.

Towards this goal, Newman and Sohler~\cite{NewmanS13} proved a general result,
showing that every property of planar graphs (and, more generally, of 
\emph{hyperfinite} graphs) is
constant-query testable. 
This theorem is a \emph{meta-theorem} in the sense that it treats many properties simultaneously.

In this paper we are interested in the \emph{running time}, which is measured as usual in terms of $n$.
For many individual properties is it easy to see that the query complexity of a tester dominates its running time. 
However, while
a bound on the running time is also a bound on the query complexity, the converse is not true in general.
 Indeed, using the Newman-Sohler theorem it is not hard to see that 
there are undecidable properties  
that are constant query testable, as is noted in~\cite{AdlerH18};
the idea is to let the number of vertices of edgeless graphs
represent the Goedel numbering of Turing machines, and to let the property be the class of 
graphs representing a Turing machine that halts on empty input. Since this property is a
class of edgeless graphs it is hyperfinite and thus constant-query testable.

In practice, \emph{sublinear running time is highly desirable}. This motivates our aim of providing 
meta-theorems for property testing with constant \emph{running time}. 
Towards this we use the framework of logic, which allows to treat many properties simultaneously. 
First-order logic on graphs $(\fo)$ speaks about the edge relation and allows quantification over vertices as usual.
The logic $\fomod$ extends $\fo$ by additional modulo counting quantifiers. Monadic second-order logic $(\mso)$
extends $\fo$ by quantification over subsets of vertices, and $\cmso$ extends $\mso$ by modulo counting quantifiers.
The following result was shown in~\cite{AdlerH18}. (Indeed, it was shown in the slightly more general setting of
relational structures.)

\begin{theorem}[\cite{AdlerH18}]\label{theo:AH18}
	Let $d,w\in \mathbb N$. On classes $\mathcal C^{w}_d$ containing all graphs with both 
	degree at most $d$ and 
	treewidth at most $w$ every problem expressible in $\cmso$
	is testable with \emph{polylogarithmic} 
	running time.
\end{theorem} 
\noindent In~\cite{AdlerH18} it was stated as a open problem whether this can be improved to \emph{constant} running time.\footnote{More precisely, this
was asked for relational structures instead of graphs.}
In this paper we approach this question by considering \emph{finitary} graph classes 
(\ie graph classes with a uniform upper bound on the component size).

\paragraph*{Results.}
By the classical Feferman-Vaught Theorem~\cite{Feferman1959} it is known that satisfaction 
of a first-order formula only depends on the formulas that are satisfied by the connected 
components. Our contribution goes further by showing constant time testability.

For fixed integers $c,d\in\mathbb N$, let $\mathcal \CC^c_d$ be the class of all graphs whose connected components contain at most $c$ vertices
and whose degree is bounded by $d$. 
Our main result is the following (cf.~Theorem~\ref{thm:testingFO+MOD}). 

\medskip\noindent 
\emph{Every property definable in $\fomod$ is testable on $\mathcal \CC^c_d$ with 
\emph{constant} running time in the bounded degree model.}

\medskip\noindent 
It is known that $\fo$ and $\mso$ have the same expressive power on finitary graph classes~\cite{ElberfeldGT16}, and
we obtain the following corollary  (cf.~Corollary~\ref{cor:testingCMSO}).

\medskip\noindent 
\emph{Every property definable in $\cmso$ is testable on $\mathcal \CC^c_d$ with 
\emph{constant} running time in the bounded degree model.}

\medskip\noindent  
While finitary classes are quite restricted, the results are interesting for several reasons.  
Firstly, finitary graph classes 
are exacly the classes of graphs that both have bounded \emph{treedeph}\footnote{Since we do not need \emph{treedepth} explicitly in this paper, we omit the definition. It can \eg be found in
~\cite{ElberfeldGT16}.} and bounded degree.
So this is a natural step towards resolving the open problem of~\cite{AdlerH18} stated above.
Moreover,
finitary graphs are at the core
of the definition of \emph{hyperfinite}\footnote{Since we do not need \emph{hyperfinite} graphs explicitly in this paper, we omit the definition. 
It can be found in~\cite{NewmanS13}.} graphs. Hence our result is an important
corner stone on the way to exploring constant time testability of logics on larger graph classes.

It is known that all monotone, 
hyperfinite 
graph properties are testable with constant running time~\cite{BenjaminiSchrammShapira2010}, as well as hereditary, hyperfinite properties~\cite{CzumajSS09}
(this follows from the proofs). Here \emph{monotone} means closed under taking subgraphs and \emph{hereditary} means
closed under taking induced subgraphs. It is easy to see that there are properties definable in $\fo$ that
are neither hereditary nor monotone. In fact our motivating example (cf.\ Example~\ref{exampleformula}) is such a property.

\paragraph*{Techniques.}
Our proof techniques rely on a version of Hanf normal form for $\fomod$ formulas, which is based on
a result by Nurmonen~\cite{Nurmonen00} and was further investigated in~\cite{HeimbergKS16}. We carefully tailor it to our classes 
$\mathcal \CC^c_d$ in such a way that it can be translated into speaking about isomorphism types of connected components.
This allows us to reduce testing an $\fomod$-property $\phi$ to testing for satisfaction of a
\emph{$k$-capped component histogram vector with modulo counting}. Such a histogram vector $\overline v$ is indexed by the
isomorphism types of connected components of size at most $c$. 
For each such component type $t$
the entry in $\overline v[t]$ is either an integer between $0$ and $k-1$ (\emph{rare} types), or it is of the 
form `$\geq k$ and congruent $i$ mod $m$' (\emph{frequent} types).
Here $k,i,m\in \mathbb N$ are natural numbers with $i<m$, that only depend on $\phi$. We say that a graph $G\in \mathcal C$ 
\emph{satisfies} $\overline v$ if
for each isomorphism type $t$, the number of occurrences of components in $G$ of type $t$ obeys the entry 
$\overline v[t]$ of the histogram vector $\overline v$.
Given access to an input graph $G\in\mathcal C$ on $n$ vertices (which we may assume is large), our algorithm
picks a constant number of vertices in $G$ at random, explores their surrounding connected components
and thus discovers a constant size subgraph $G'$ of $G$. 
If a component of $G'$ is rare according to $\overline v$, then the algorithm rejects. Otherwise it does an
 arithmetic check to determine whether the graph is close to satisfying $\phi$.
In the correctness proof, we check closeness by trying to `plant' components as required by the rare components and the remainders in 
$\overline v$, then we fill the remaining of the $n$ vertices with copies of a graph similar to $G'$ 
in order to preserve satisfaction of $\overline v$.
Finally, there may be a small number of vertices left to fill and we check for consistent `patchability' using
arithmetic techniques including
the Frobenius Coin Theorem (cf.~\eg\,\cite{Kannan92}).

\paragraph*{Further related work.}
In~\cite{AdlerF23} a slight relaxation of the bounded degree model was introduced, where the `closeness' measure of graphs
allows a small fraction of the \emph{vertices}
to be removed or added, in addition to the usual edge modifications. For this model, a theorem similar to Theorem~\ref{theo:AH18} 
was proven, however, with \emph{constant} running time.
This suggests that the intricacy arises when the number of vertices cannot be modified. 
This is also the case in our setting. In particular, our main algorithm (cf.~Theorem~\ref{thm:testingFO+MOD}) solves the 
issue of understanding from local information only, whether the property can be realised on the given number 
$n$ of input vertices.

In~\cite{AdlerKP21,AdlerKP24}, testability of $\fo$ on general bounded degree graphs was studied and it was shown that 
the fragment $\Sigma_2$ (\ie the fragment of $\fo$ of formulas in the $\exists^*\forall^*$-prefix class) is testable with constant
query complexity in the bounded-degree model, whereas there is a property in $\Pi_2$ (the $\forall^*\exists^*$-prefix class)
that is not testable with constant query complexity in this model.
In the dense graph model, a similar picture is known~\cite{AlonFischerKrivelevichSzegedy2000}. 
However, the proof techniques differ fundamentally.

\paragraph*{Structure of the paper.}
We begin with the basic definitions and properties in Section~\ref{sec:prelim}.
In Section~\ref{sec:hanf} we derive our tailored version of Hanf normal form for $\fomod$ and show that it can be equivalently represented by a finite number of
capped component histogram vector with modulo counting.
Our main results are presented in Section~\ref{sec:main}.
We conclude in Section~\ref{sec:conclusion}.

\section{Preliminaries}\label{sec:prelim}

We will use $\mathbb N$ to denote the natural numbers including $0$ and $\Z^+:=\mathbb N\setminus\{0\}$. We let
$[n]:=\{1,\ldots,n\}$.

\paragraph*{Graphs}

A graph is a pair $G=(V(G),E(G))$ where $V(G)$ is a finite sets and $E(G)\subseteq \{e\in\P(V(G))\mid \abs{e}=2\}$. Elements of $V(G)$ are called \emph{vertices} (of $G$) and elements of $E(G)$ are called \emph{edges} (of $G$). An edge $\{x,y\}$ may also be denoted $xy$. If $xy\in V(G)$ we say that $x$ and $y$ are \emph{adjacent} or that $x$ is a \emph{neighbour} of $y$. A graph $H$ is a \emph{subgraph} of a graph $G$ if $V(H) \subseteq V(G)$ and $E(H) \subseteq E(G)$. 
For a graph $G$ and $S\subseteq V(G)$ we let $G[S]:=(S, \{\{v,w\}\in E(G): v,w\in S \})$ be the subgraph of $G$ \emph{induced} by $S$. 
A subgraph $H$ of $G$ is an \emph{induced subgraph} of $G$ if $H=G[S]$ for some set $S\subseteq V(G)$. For vertices $v,w\in V(G)$ we say that $v$ is a \emph{neighbour} of $w$ if $\{v,w\}\in E(G)$.
A \emph{path} of length $n$ from $p_0$ to $p_n$ is a graph $P=(V(P),E(P))$ with $V(P)=\{p_0,p_1,\dots,p_n\}$ and $E(P)=\{p_0p_1,p_1p_2,\dots,p_{n-1}p_{n}\}$ such that $p_0\neq p_n$. A subset $X\subseteq V(G)$ is called \emph{connected} if for every pair of distinct vertices $x,y\in X$, $G$ contains a path from $x$ to $y$ as a subgraph. The \emph{distance} $\dist_G(x,y)$ between vertices $x$ and $y$ in $G$ is the length of the shortest path from $x$ to $y$.
For a graph $G$, a vertex $v\in V(G)$ and $r\in \mathbb N$ we let $N_G(v,r):=\{u\in V(G)\mid \dist_G(v,u)\leq r\}.$ 
A subset $X\subseteq V(G)$ is a \emph{connected component} of $G$ if it is connected and no set $Y$ satisfying $X \subsetneq Y \subseteq V(G)$ is connected. A graph $G$ is connected if $V(G)$ is connected.
An \emph{isomorphism} from graph $G$ to graph $H$ is a bijective function $f:V(G)\to V(H)$ such that $\{u,v\}\in E(G)$ if and only if $\{f(u),f(v)\}\in E(H)$. Two graphs are called \emph{isomorphic} if there exists an isomorphism between them. A \emph{graph property} is a class of graphs which is closed under isomorphism. We say a graph $G$ has property $P$ if $G\in P$. If $P$ is a property and $C$ is a class of graphs such that $P\subseteq C$, we say that $P$ is a \emph{property on} $C$. A \emph{monotone} property is a property closed under taking subgraphs. A \emph{hereditary} property is a property closed under taking induced subgraphs.

The \emph{degree} of $v\in V(G)$ is the number of vertices adjacent to $v$ in $G$. The degree of a graph $G$ is the maximum degree of a vertex in $V(G)$. A class of graphs has $d$-bounded degree if every graph in the class has degree at most $d$. The class of all $d$-bounded degree graphs is denoted $\CC_d$.

A \emph{rooted graph} is a pair $(G,v_G)$ where $G$ is a graph $v_G\in V(G)$ is the \emph{root}. An \emph{isomorphism} from rooted graph $G$ to rooted graph $H$ is an isomorphism $f:V(G)\to V(H)$ such that $f(v_G) = v_H$. 
The \emph{radius} of a rooted graph is the smallest $r\in\N$ such that such that each vertex has distance at most $r$ from the root. A rooted graph of radius $r$ is also called an \emph{$r$-ball}.
Let $G$ be a graph and $v\in V(G)$. 
The \emph{$r$-ball} in $G$ \emph{around} $v$ is the rooted graph $\B_G(v,r):=(G[N_G(v,r)]),v)$. 
The number of possible non-isomorphic $r$-balls in $d$-bounded degree graphs is finite and denoted by $N(r,d)$. 
The set of \emph{isomorphism types} of $r$-balls in $d$-bounded degree graphs is denoted by $\Tball(r,d)\coloneq\{\tau_1,\tau_2,\dots,\tau_{N(r,d)}\}$. Note that if $r_1<r_2$, then $\Tball(r_1,d) \subset \Tball(r_2,d)$.

The \emph{ball histogram vector} $\bhv_r(G)$ of a $d$-bounded degree graph $G$ is a vector of length $N(r,d)$ whose $i$th entry is the number of vertices $v$ of $G$ such that $\B_G(v,r)$ has isomorphism type $\tau_i$. The \emph{ball distribution vector} $\bdv_r(G)$ is obtained from $\bhv_r(G)$ by dividing each entry of $\bhv_r(G)$ by $\abs{V(G)}$.
For a vector $\ol{u}$ of length $n$ we define $\cab{\ol{u}} \coloneq \sum_{i=1}^{n} \abs{\ol{u}[i]}$.
 
For a graph class $\mathcal C$ and a constant $c\in \mathbb N$ we say that $\mathcal C$ 
is \emph{$c$-finitary}, if all connected components of the graphs in $\mathcal C$ have size at most $c$.
We say that a class $\mathcal C$ is \emph{finitary}, if there is a $c\in\mathbb N$ such that $\cal C$ is $c$-finitary. The class of all $c$-finitary graphs is denoted $\CC^c$. Further $\CC^c_d \coloneq \CC^c \cap \CC_d$.

The number of possible non-isomorphic connected components of $d$-bounded degree, $c$-finitary graphs is finite and denoted $M(c,d)$. The set of $d$-bounded degree, $c$-finitary \emph{component (isomorphism) types} is denoted $T_\textnormal{comp}(c,d)\coloneq\{t_1,t_2,\dots,t_{M(c,d)}\}$. For a component type $t$, $\abs{t}$ denotes the number of vertices of any graph of that type.

The \emph{component histogram vector} $\chv_c(G)$ of a $d$-bounded degree, $c$-finitary graph $G$ is a vector of length $M(c,d)$, whose $i$th entry is the number of connected components of $G$ which are isomorphic to $t_i$. Let $k\in\N$. A \emph{$k$-capped component histogram vector} ($k$-CCHV) $\ol{a}$ is a vector of length $M(c,d)$ with entries from the set 
$\{0,1,2,\dots,k-1\}\cup\{(j,\ell)\mid j,\ell \in \mathbb N\text{ and }j<\ell\}$. A $d$-bounded degree, $c$-finitary graph $G$ is said to \emph{satisfy} a $k$-capped component histogram vector $\ol{a}$ if for all $ i \in\{1,\dots,M(c,d)\}$, either $\ol{a}[i]=\chv_c(G)[i]$, or $\ol{a}[i]=(j,\ell)$, $\chv_c(G)[i]\geq k$ and 
$\chv_c(G)[i]\equiv j(\Mod\ell)$. For a $k$-CCHV $\ol{a}$ and component type $t$, we may use the notation $\ol{a}[t]$ to denote $\ol{a}[i]$ where $i\in\{1,\dots,M(c,d)\}$ is such that $t=t_i$.
For each ball type $\tau\in\Tball(r,d)$ with $c$ vertices we denote by $\underline{\tau}$ its \emph{underlying} component type $\underline{\tau}\in\Tcomp(c,d)$, the component type isomorphic to an $r$-ball of type $\tau$ with the root forgotten.

\begin{remark} Let $c,d\in\N$ and let $G$ be a $d$-bounded degree, $c$-finitary graph. Algorithm~\ref{algo:histconversion}
	computes $\chv_c(G)$ from $\bhv_c(G)$ in constant time.
\end{remark}
	\begin{algorithm}
		\tbf{input} ball histogram vector $\bhv_c(G)$\;
		$\chv_c(G) \gets 0$\;
		\For{$j\in [M(c,d)]$}{
			choose a representative $H\in t_j$\;
			choose any $v\in H$\;
			let $\tau_i$ be the isomorphism type of $\B_{H}(v,c)$\;
			$\chv_c(G)[j] \gets \frac{\bhv_c(G)[i]}{\bhv_c(H)[i]}$\;
		}
		\Return $\chv_c(G)$
		\caption{$\bhv_c$ to $\chv_c$}\label{bhvtochv}
		\label{algo:histconversion}
	\end{algorithm}

\paragraph*{Logic}
\emph{First-order logic} (FO) formulas on graphs are built from predicates for the edge relation and equality, using Boolean connectives $\land,\lor,\neg$ as well as existential and universal quatifiers $\exists,\forall$, where the variables represent graph vertices. A \emph{sentence} of FO is a formula in which every variable is in the scope of a quantifier. The semantics of FO are defined in the usual way. If a graph $G$ satisfies a sentence $\phi$ we write $G\models\phi$. A sentence $\phi$ of FO \emph{defines} a graph property $P_\phi$ if the members of the property are exactly the graphs which satisfy $\phi$.

\begin{example}\label{exampleformula}
The formula \[\phi=\exists x \Bigl(\forall y \bigl(\neg Exy\bigr)\Bigr) \land \exists x\exists y\Bigl(Exy \land \forall z \bigl((Exz \to z=y) \land(Eyz \to z=x)\bigr)\Bigr)\] defines the property of graphs which contain both an isolated vertex and an isolated edge as connected components. Its negation $\neg\phi$ will be our running example.
\end{example}

\emph{First-order logic with counting} ($\fomod$) extends $\fo$ by modular counting quantifiers $\exists^{j(\mod m)}x$
for a variable $x$ and $j,\ell\in\N$ and $j<\ell$.

\emph{Monadic second-order logic} ($\mso$) extends $\fo$ by also additionally allowing quantification over subsets of vertices.
Formally, we have two types of variables, \emph{individual variables} (denoted by small letters $x,y,z,x_1,\ldots$) which are interpreted by vertices, and \emph{set variables} (denoted by capital 
letters $X,Y,Z,X_1,\ldots$) which are interpreted by subsets of vertices. In addtion to the atomic first-order formula we also have the formula $Xx$ saying $x$ is an element of set $X$. Furthermore, we
have universal and existential quantification both over individual variables and set variables.
\emph{Monadic second-order logic with counting} ($\cmso$) extends $\mso$ by modular counting quantifiers $\exists^{i(\mod m)}x$
for an individual variable $x$ and $j,\ell\in\N$ and $j<\ell$.

\paragraph*{Property testing}
We will use the bounded-degree model for property testing introduced in~\cite{GoldreichR02}. A property testing algorithm $\mathcal{A}$ does not have unrestricted access to the input graph $G$. Instead it recieves $n\coloneq\abs{V(G)}$ and the labels $1,\dots,n$ of the vertices of $G$. We assume $\mathcal{A}$ may make queries to an \emph{oracle} which on a query $(i,j)$ for $i\in\{1,\dots,n\}$ and $j\in\{1,\dots,d\}$, returns the $j$th neighbour of the vertex $i$, or $\perp$ if $i$ has less than $j$ neighbours.
We assume that oracle queries are answered in constant time. 
An \emph{edge modification} to a graph $G$ is the addition or deletion of an edge between any pair of vertices in $V(G)$. The \emph{distance} $\dist(G,H)$ between graphs $G$ and $H$ with equal numbers of vertices is the minimum number of edge modifications needed to make $G$ isomorphic to $H$.
Let $0<\eps\leq 1$. Graphs $G$ and $H$ with $n$ vertices and $d$-bounded degree are called \emph{$\eps$-close} if $\dist(G,H)\leq\eps d n$. Otherwise they are called \emph{$\eps$-far}.
We say that a graph $G$ is $\eps$-close to a property $P$ if it is $\eps$-close to some graph $H\in P$. Otherwise it is $\eps$-far from $P$.
If $\ol{u_1}$ and $\ol{u_2}$ are distribution vectors 
of the same length, we call them \emph{$\eps$-close} if $\cab{\ol{u_1}-\ol{u_2}}\leq\eps$ otherwise we call them \emph{$\eps$-far}.

\begin{definition}
Fix $d\in \mathbb N$, let $0<\eps\leq 1$ and let $P$ be a property on $\mathcal C_d$, \ie $P\subseteq \mathcal C_d$. 
	An \emph{$\eps$-tester} for $P$ is an algorithm with input $n=\abs{V(G)}$ and oracle access to $G\in \mathcal C_d$ which does the following.
\begin{enumerate}
	\item If $G\in P$, accepts with probability $\geq\frac{2}{3}$
	\item If $G$ is $\eps$-far from $P$, rejects with probability $\geq\frac{2}{3}$.
\end{enumerate}
\end{definition}

The \emph{query complexity} of an $\eps$-tester is the maximum number of queries the tester makes to the oracle as a function of $n$. A property $P$ is called \emph{testable} if for every $0<\eps\leq1$ and $n\in\N$ there exists an $\eps$-tester for the property $P_n\coloneq\{G\in P \mid \abs{G}=n\}$ with constant query complexity. 

	We are interested in $\eps$-testers that are uniform in $n$ and in their \emph{running time}, where the running time 
	is defined in the usual way as a function of $n$. 
	A property $P$ is (uniformly) \emph{testable with constant} running time, if for every $0<\eps\leq1$ there exists an $\eps$-tester for the property $P$
	with constant running time (and hence constant query complexity).

\paragraph*{Model of computation.} We use Random Access Machines (RAMs) and a uniform cost measure
when analysing our algorithms; \ie we assume all basic arithmetic operations including random
sampling can be done in constant time, regardless of the size of the numbers involved.

We consider a small example.

\begin{example}\label{ex:closeness}
Let $\mathcal C^2_1$ be the class of all $2$-finitary graphs with degree at most $1$ and let
	$\psi=\neg\phi$, where $\phi$ is the $\fo$ sentence in Example \ref{exampleformula}. The property $P_\psi$ contains exactly the graphs which do \tbf{not} contain both an isolated vertex and an isolated edge. Consider the graphs
	$G^{2n+1}$ and $G^{2n}$ on $2n+1$ and $2n$ vertices, respectively, shown in Figure~\ref{fig:test}. 
	Clearly $G^{2n}\models\psi$ and $G^{2n+1}\not\models\psi$. Moreover, for sufficiently small $\epsilon$,  $G^{2n+1}$ is $\epsilon$-far 
from satisfying $\psi$ as in order to satisfy $\psi$ every edge must be deleted.\footnote{Here one might argue that
	farness of $G^{2n+1}$ is an artificial artefact due to the choice of the model, as we are not allowed to simply add an edge and exceed the degree bound. 
However, we could always add a sentence
to our formula imposing the degree of at most $1$.} Moreover, 
with high probability sampling a constant number of the vertices of either graph and querying their neighbourhoods will only detect isolated edges, 
	so sampling alone cannot distinguish $G^{2n+1}$ from $G^{2n}$. Hence the tester will have to make use of knowing the number of
	vertices.

	We will show constant time testability of $P_\psi$ on $\mathcal C^2_1$ in Example~\ref{ex:final}, illustrating the proof techniques of our main theorem. 	
	Note that $P_\psi$ is neither monotone nor hereditary, so constant time testability 
	does not follow from known results~\cite{BenjaminiSchrammShapira2010,CzumajSS09}.
	
\end{example}

\begin{figure}
\begin{center}
\setlength{\fboxsep}{0pt}
\fbox{\begin{tikzpicture}[vnode/.style={circle,draw=black,fill=black,minimum size=1mm,inner sep=2pt}]
	\node (0) at (3.5,-1) {$G^{2n+1}$};
	
	\node at (0.5,0) {};

	\node[vnode] (1) at (1,0) {};

	\node[vnode] (3) at (2,0.4) {};
	\node[vnode] (4) at (2,-0.4) {};
	\node[vnode] (5) at (3,0.4) {};
	\node[vnode] (6) at (3,-0.4) {};
	
	\draw[dotted,thick] (4,0)--(5,0);
	
	\node[vnode] (7) at (6,0.4) {};
	\node[vnode] (8) at (6,-0.4) {};
	
	\draw[thick] (3)--(4);
	\draw[thick] (5)--(6);
	\draw[thick] (7)--(8);
	
	\draw[dotted,thick] (4,0)--(5,0);
	
	\draw (7,1) -- (7,-1.3);
	
	\node (zero) at (10.5,-1) {$G^{2n}$};
	
	\node at (13.5,0) {};

	\node[vnode] (a) at (8,0.4) {};
	\node[vnode] (b) at (8,-0.4) {};
	\node[vnode] (c) at (9,0.4) {};
	\node[vnode] (d) at (9,-0.4) {};
	\node[vnode] (e) at (10,0.4) {};
	\node[vnode] (f) at (10,-0.4) {};
	
	\draw[dotted,thick] (11,0)--(12,0);
	
	\node[vnode] (g) at (13,0.4) {};
	\node[vnode] (h) at (13,-0.4) {};
	
	\draw[thick] (a)--(b);
	\draw[thick] (c)--(d);
	\draw[thick] (e)--(f);
	\draw[thick] (g)--(h);
\end{tikzpicture}}
\end{center}
\caption{The graphs $G^{2n+1}$ and $G^{2n}$ of Example~\ref{ex:closeness}.}
\label{fig:test}
\end{figure}
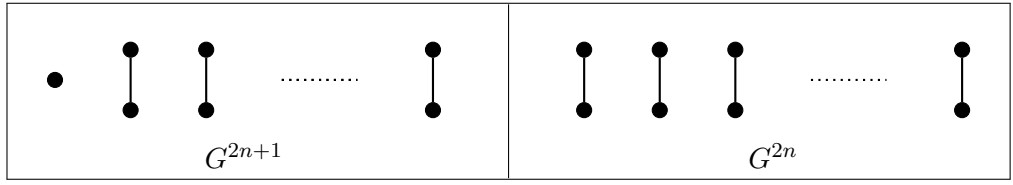

\smallskip
We will use the following two lemmas from~\cite{NewmanS13}.
\begin{lemma}[\cite{NewmanS13}]\label{freqclosegraphclose} 
Let $c\in\Z^+$ and $0<\lambda<1$. Let $G_1$ and $G_2$ be $c$-finitary graphs on $n$ vertices with degree bound $d$.
If $\cab{\bdv_c(G_1)-\bdv_c(G_2)}\leq\lambda$, then $G_1$ and $G_2$ are $\lambda$-close.
\end{lemma}

	Let $r,s\in\Z^+$. For an input graph $G$ of degree at most $d$, Algorithm \ref{estfreq} samples $s$ vertices of $G$ uniformly at random and returns an estimation of the frequencies of $r$-ball types in $G$.

\smallskip
\begin{algorithm}[H]
\caption{\textsc{EstimateFrequencies}($G$,$r$,$s$)}\label{estfreq}
	\smallskip
	$\widetilde{\bdv}_r(G) \gets 0$\;
	sample $s$ vertices $u_1,\dots,u_s$ uniformly at random\;
	\For{$j\in\{1,\dots,s\}$}{
		explore the $r$-ball at $u_j$\;
		let $i$ be the index of the isomorphism type $\tau_i$ which is isomorphic to this $r$-ball\;
		$\widetilde{\bdv}_r(G)[i] \gets \widetilde{\bdv}_r(G)[i]+\frac{1}{s}$\;
	}
	\Return $\widetilde{\bdv}_r(G)$
\end{algorithm}

\begin{lemma}[\cite{NewmanS13}]\label{freqestlemma} 
For a $d$-bounded degree graph $G$, $r,q\in\Z^+$ and $0<\lambda<1$, if
\[q\geq\frac{N(r,d)^2}{\lambda^2}\cdot\ln(N(r,d)+40),\]
	then the vector $\widetilde{\bdv}_r(G)$ returned by \textsc{EstimateFrequencies}$(G,r,q)$ 
	satisfies that \linebreak $\cab{\widetilde{\bdv}_r(G) - \bdv_r(G)}\leq\lambda$ with probability $19/20$.
\end{lemma}

\paragraph*{Frobenius Coin Theorem}

We will apply the following results of number theory to the question of whether a graph on a given number of vertices can be created as a disjoint union of pre-described connected component types. A \emph{conical combination} of integers $a_1,\dots,a_n$ is a sum $\sum_{i=1}^{n}a_i b_i$ for some natural numbers $b_1,\dots,b_n$. A collection of integers is \emph{setwise coprime} if their greatest common divisor is 1.
For integers $a_1,\dots,a_n$ we write $\gcd(a_1,\dots,a_n)$ and $\lcm(a_1,\dots,a_n)$ to denote their greatest common divisor and lowest common multiple respectively.

\begin{theorem}[\cite{Kannan92}]\label{fct}
For every setwise coprime subset of the positive integers $\{a_1,\dots,a_n\}$ there exists a largest integer which is not expressible as a conical combination of $a_1,\dots,a_n$. This integer is called the \textit{Frobenius number}.
\end{theorem}

\begin{corollary}\label{fct+}
There exists a function $F:\P(\Zpos) \to \Zpos$ which maps $\{a_1,\dots,a_n\}$ to the largest multiple of $g\coloneq\gcd(a_1,\dots,a_n)$ which is not expressible as a conical combination of $a_1,\dots,a_n$.
\begin{proof}
Notice that if $\gcd(a_1,\dots,a_n)=g$, then the integers $\frac{1}{g}a_1,\dots,\frac{1}{g}a_n$ are setwise coprime and thus by the previous theorem there exists a largest integer $x$ not expressible as
$\sum_{i=1}^{n}b_i\left(\frac{a_i}{g}\right)$
	for any $b_1,\dots,b_n\in\N$. It follows that $g x$ is \emph{not} expressible as a conical combination of $a_1,\dots,a_n$. 
	Further for any integer $y>x$, the product $g y$ \emph{is} expressible as a conical combination of $a_1,\dots,a_n$ as 
	 $y>x$ is a conical combination of $\frac{1}{g}a_1,\dots,\frac{1}{g}a_n$.
\end{proof}
\end{corollary}

\section{From Hanf normal form to component histogram vectors}\label{sec:hanf}

Let $r\in\N$ and $\tau\in T(r,d)$. There is an $\fomod$ formula $\ball^r_\tau(x)$ such that $G \models \ball^r_\tau(v)$ iff $\B_G(v,r)$ is isomorphic to $\tau$. A \emph{Hanf sentence} is a sentence of either the form

\[\exists^{\geq m}x(r,\tau) \coloneq \exists x_1 \exists x_2, \dots, \exists x_m \left(\bigland_{1 \leq i<j\leq m} x_i \neq x_j \land \bigland_{1 \leq i \leq m} \ball^r_\tau(x_i)\right)\]
or the form
$\exists^{j(\Mod \ell)}x(r,\tau)\coloneq \left(\exists^{j(\Mod \ell)}x \left(\ball^r_\tau(x)\right) \right)$

for some $r\in\N$ and $\tau\in \Tball(r,d)$.
We also define the shortcut 

$\exists^{=m}x(r,\tau) \coloneq \left(\exists^{\geq m}x(r,\tau)) \land \neg(\exists^{\geq m+1}x(r,\tau)\right)$.

A sentence of FO is said to be in \emph{Hanf normal form} if it is a Boolean combination of Hanf sentences.
The \emph{radius} and \emph{type} of the sentence $\exists^{\geq m}x(r,\tau)$, $\exists^{j(\Mod \ell)}x(r,\tau)$ or $\exists^{=m}x(r,\tau)$ are $r$ and $\tau$ respectively.
For FO-sentences $\phi$ and $\psi$ we write $\phi \equiv_d \psi$ if $P_\phi \cap \CC_d = P_\psi \cap \CC_d$ and $\phi \equiv^c_d \psi$ if $P_\phi \cap \CC^c_d = P_\psi \cap \CC^c_d$.

\begin{theorem}[Hanf's Locality Theorem for $\fomod$ \cite{HeimbergKS16}]\label{hanflocalityfo}
Let $d\in\N$. For every sentence $\phi$ of $\fomod$ and every $d\in\N$ there is an $\fomod$ sentence $\psi$ in Hanf normal form such that $\phi \equiv_d \psi$.
\end{theorem}

We tailor Hanf's theorem to our needs, towards expressing $\fomod$ sentences on finitary graphs as collections of CCHVs. The following lemma starts this process by unifying the radii of the Hanf sentences.

\begin{lemma}\label{hnfuniformradius}
Let $d\in\N$. For every sentence $\phi$ of $\fomod$ there exist an $\rmax\in\N$ and an $\fomod$ sentence $\psi$ such that $\phi \equiv_d \psi$, where $\psi$ is in HNF and every Hanf sentence has radius $\rmax$.
\begin{proof}
See appendix \ref{app:hnfuniformradius}.
\end{proof}
\end{lemma}

For our tailored version of HNF it is also necessary for the ``cap'' of every ``at least'' quantifier to be the same over all Hanf sentences. The following lemma achieves this.

\begin{lemma}\label{hnfuniformcap}
Let $d\in\N$. For every sentence $\phi$ of $\fomod$, there exist a $k\in\N$ and an $\fomod$ sentence $\psi$ such that $\phi \equiv_d \psi$, where
$\psi$ is in disjunctive normal form (DNF) of atoms of one of the forms
\begin{itemize}
	\item $\exists^{\geq k}x(r,\tau)$ for some $r\in\N$ and $\tau\in\Tball(r,d)$
	\item $\exists^{=m}x(r,\tau)$  for some $m\in\N$ s.t. $m<k$, $r\in\N$ and $\tau\in\Tball(r,d)$
	\item $\exists^{j(\Mod \ell)}x(r,\tau)$  for some $j,\ell\in\N$ s.t. $j<\ell$, $r\in\N$ and $\tau\in\Tball(r,d)$.
\end{itemize}
In particular these literals may only have radii and types from among those of Hanf sentences of $\phi$.
\begin{proof}
See appendix \ref{app:hnfuniformcap}.
\end{proof}
\end{lemma}

For $\tau\in\Tball(c-1,d)$, let the \emph{repetition} $\rep(\tau)$ be the number of vertices $v$ with $B_H(v,c-1)=\tau$ in an arbitrary representative $H\in\underline{\tau}$.

\begin{corollary}\label{fohistsentence}
Let $c,d\in\N$, $c\geq1$. For every sentence $\phi$ of $\fomod$ there exist an $\fomod$ sentence $\psi$ and $k\in\N$ such that $\phi \equiv^c_d \psi$ and $\psi$ is a DNF of atoms of one of the forms
\begin{itemize}
	\item $\exists^{\geq k\cdot\rep(\tau)}x(c-1,\tau)$ for some $\tau\in\Tball(c-1,d)$
	\item $\exists^{=m}x(c-1,\tau)$  for some $m\in\N$ s.t. $m<k\cdot\rep(\tau)$ and $\tau\in\Tball(c-1,d)$
	\item $\exists^{j(\Mod \ell)}x(c-1,\tau)$  for some $j,\ell\in\N$ s.t. $j<\ell$ and $\tau\in\Tball(c-1,d)$.
\end{itemize}
\begin{proof}
See appendix \ref{app:fohistsentence}.
\end{proof}
\end{corollary}

\begin{lemma}\label{fohists}
Let $c,d\in\N$, $c\geq1$. For every satisfiable $\fomod$ sentence $\phi$ there exist a $k\in\N$ and a finite 
set $X$ of $k$-capped component histogram vectors such that $P_\phi \cap \CC^c_d$ contains exactly the 
$c$-finitary, $d$-bounded degree graphs which satisfy at least one of the histogram vectors in $X$, i.e. \[P_\phi \cap \CC^c_d = \{G\in \CC^c_d \mid G\textnormal{ satisfies $\ol{a}$ for some $\ol{a}\in X$}\}.\]
\begin{proof}
See appendix \ref{app:fohists}.
\end{proof}
\end{lemma}

\begin{example}\label{ex:final}
	On the class $\mathcal C^2_1$, the number of component types is $M(2,1)=2$ and the  
	 component types are $t_1$ and $t_2$, where $t_1$ is an isolated vertex (a vertex with no neighbours) and $t_2$ is 
	 an isolated edge (a pair of vertices connected by an edge which have no other neighbours).

	Let $\psi$ be the $\fo$ sentence in Example~\ref{ex:closeness}. The property $P_\psi$ contains exactly the graphs which do \tbf{not} contain both type $t_1$ and type $t_2$ connected components. By Corollary \ref{fohists} we know that for some $k$ there is a collection of $k$-capped component histogram vectors such that every graph in $P_\psi$ satisfies at least one of them. It is easy to see that the 1-CCHVs $(0,0)$, $(0,(0,1))$, and $((0,1),0)$ meet this condition.

\medskip
We show that $P_\psi$ is uniformly testable in constant time. Let $0<\eps\leq1$ and assume we have oracle access to an input graph $G$ on $n$ vertices where $n\geq n_0$
	for some sufficiently large constant $n_0$.
(We can use an exact algorithm to decide the problem if $n\leq n_0$ for some constant $n_0\in \mathbb N$.)

The following is a constant time $\eps$-tester for $P_\psi$. Fix \[s = \ceil{\frac{N(2,1)^2}{\eps^2}\cdot\ln(N(2,1) + 40)}.\]
Call the algorithm \textsc{EstimateFrequencies}$(G,,s)$ to compute $\widetilde{\bdv}_c(G)$. From this we compute $\widetilde{\chv}_2(G)$, 
the histogram vector of the connected components seen during sampling. 
Let $n_1=\widetilde{\chv}_2(G)[1]$ and $n_2=\widetilde{\chv}_2(G)[2]$. If $n_2=0$ or ($n_1=0$ and $n$ is even), accept. Otherwise reject.

\begin{proof}
	The tester clearly proceeds in constant time.
	Suppose $G$ satisfies $P_\psi$. Then either $E(G)=\emptyset$ or $G$ is a disjoint union of isolated edges (and thus $|V(G)|$ is even). 
	In both cases the algorithm accepts.
	Suppose $G$ is $\eps$-far from satisfying $P_\psi$. Suppose towards contradiction that the tester accepts. We have two cases.

\tbf{Case 1.} No $t_2$ component was seen during sampling. Let $G'$ be the graph with $n$ vertices and no edges. 
	Clearly $G'\in P_\psi$. Since $\bdv_2(G') = (1,0) =\widetilde{\bdv}_2(G)$, we can apply Lemma \ref{freqestlemma} to see that with probability $\frac{19}{20}$, $\lVert \bdv_2(G')-\bdv_2(G) \rVert_1 = \cab{\widetilde{\bdv}_2(G) - \bdv_2(G)} \leq \epsilon$ and applying Lemma \ref{freqclosegraphclose} we see $G'$ and $G$ are $\eps$-close, a contradiction.

	\tbf{Case 2.} No $t_1$ component was seen during sampling and $G$ has an even number of vertices. Let $G'$ be the graph on $n$ vertices which only contains $t_2$ components. Clearly $G'\in P_\psi$. Since $\bdv_2(G') = (0,1) =\widetilde{\bdv}_2(G)$, we can apply Lemma \ref{freqestlemma} to see that with probability $\frac{19}{20}$, 
	$\lVert \bdv_2(G')-\bdv_2(G) \rVert_1 = \cab{\widetilde{\bdv}_2(G) - \bdv_2(G)} \leq \epsilon$ and applying Lemma \ref{freqclosegraphclose} we see $G'$ and $G$ are $\eps$-close, a contradiction.
\end{proof}
\end{example}

\medskip
In Section \ref{sec:main} we will extend this strategy to test FO properties on finitary graphs, using testers which first sample a finite portion and then check to see if the size of the graph allows any necessary edge modifications to be made.


\section{Testing first-order logic with modulo counting on finitary graphs}\label{sec:main}

We are now ready to prove our main theorem.
\begin{theorem}\label{thm:testingFO+MOD}
Let $c,d\in\N$. 
For every sentence $\phi$ of $\fomod$, the property $P_\phi\cap\CC^c_d$
	is uniformly testable in constant time on $\CC^c_d$. 
\end{theorem}

\begin{proof}
Let $\phi$ be a sentence of FO and $c,d\in\N$. By Lemma \ref{fohists} there exists a set $X$ of $k$-capped component histogram vectors such that $P_\phi \cap \CC^c_d = \{G\in \CC^c_d \mid G\textnormal{ satisfies $\ol{a}$ for some $\ol{a}\in X$}\}$. Since testability is closed under finite unions it is sufficient to prove the result for the case that $X=\{\ol{u}\}$.

	Based on $\ol{u}$ we partition $\Tcomp(c,d)$ into two sets $\Tcomp(c,d) = \{s_1,\dots,s_m\}\cup\{t_1,\dots,t_\ell\}$, such that
\begin{itemize}
	\item $\ol{u}[s_i] \in[k-1]$ for $i\in[m]$ ($s_i$ is a ``rare'' component type), and
	\item $\ol{u}[t_i] = (a_i,b_i)$ where $a_i,b_i\in\N$, $a_i<b_i$ for $i\in[\ell]$ ($t_i$ is a ``frequent'' component type).
\end{itemize}

	\textbf{Algorithm.} Given $\epsilon\in(0,1]$, an $\eps$-tester for $P_\phi$ proceeds as follows.	
	Assume that
	\begin{align*}
	n>n_0\coloneq\frac{4}{\eps}\Bigg(&\sum_{i=1}^{m}\ol{u}[s_i]\cdot\abs{s_i} + \sum_{i=1}^{\ell}k_i\cdot\abs{t_i}+F(b_1\cdot\abs{t_1},\dots,b_\ell\cdot\abs{t_\ell}) \\
	&+ \lcm(b_1,\dots,b_\ell)\cdot(1+d^c)\cdot\ceil{\frac{N(c,d)^2}{(\eps /2)^2}\cdot\ln(N(c,d) + 40)}\Bigg)
	\end{align*}
	where $F$ is the function defined in corollary \ref{fct+}.

	Fix \[q=\ceil{\frac{N(c,d)^2}{(\eps /2)^2}\cdot\ln(N(c,d) + 40)}.\]

	Sample $q$ vertices of $G$ uniformly at random, exploring the $c$-neighbourhood at each vertex and determine its component type. We do this by calling Algorithm \ref{estfreq}, \linebreak \textsc{EstimateFrequencies}($G$,$c$,$q$) and computing the component type each sampled vertex lies in using Algorithm \ref{bhvtochv}.

	If a component of type $s_i$ is encountered, for any $i\in[m]$, reject. Otherwise, for $i\in[\ell]$ let 
	$k_i:=\min\{p\geq k\mid p\equiv a_i(\mod b_i)\}$.
	Compute the values
	$g\coloneq\gcd(b_1\cdot\abs{t_1},\dots,b_\ell\cdot\abs{t_\ell})$  and
	$n' \coloneq n - \sum_{i=1}^{m}\ol{u}[s_i]\cdot\abs{s_i} - \sum_{i=1}^{\ell}k_i\cdot\abs{t_i}$
	and accept if $g\vert n'$, otherwise reject.
	
	\textbf{Running time.} The algorithm clearly has constant running time: The number $q=q(\epsilon,d,c)$ is constant. 
	Determining the component of each sampled vertex and computing the values $g$ and $n'$ also take constant time. 

\textbf{Correctness.}
	We may assume that $n>n_0$ (otherwise we solve the problem exactly in constant time).
Assume that $G\models\phi$. Then $G$ has $\sum_{i=1}^{m}\ol{u}[s_i]$ ``rare'' components using $\sum_{i=1}^{m}\ol{u}[s_i]\cdot\abs{s_i}$ vertices in total. 
	Since this number over vertices is
	constant, with probability $\frac{1}{n}(n-\sum_{i=1}^{m}\ol{u}[s_i]\cdot\abs{s_i})$ none of these vertices will be seen during sampling and thus we will not reject initially.

Since $G$ satisfies $\ol{u}$, $G$ is the disjoint union of
\begin{itemize}
	\item $\ol{u}[s_i]$ components of type $s_i$ for each $i\in\{1,\dots,m\}$, with $\sum_{i=1}^{m}\ol{u}[s_i]\cdot\abs{s_i}$ vertices in total,
	\item $k_i$ components of type $t_i$ for each $i\in\{1,\dots,\ell\}$, with $\sum_{i=1}^{\ell}k_i\cdot\abs{t_i}$ in total, and
	\item some multiple of $b_i$ copies of type $t_i$ for each $t_1,\dots,t_\ell$, over the remaining $n'$ vertices.
\end{itemize}

	Since $g$ divides $b_i\cdot\abs{t_i}$ for every $i\in[\ell]$, we have that $g$ divides $n'$ and thus we accept.
	Assume that $G$ is $\eps$-far from $\phi$. Towards a contradiction suppose we accept; in particular none of the component types $s_1,\dots,s_m$ are seen during sampling and $g\vert n'$. We construct a graph $G'$ such that $G'\models\phi$ and $G'$ is $\eps$-close to $G$.  
Let $G_0 \subseteq G$ be the graph consisting of all the vertices and edges seen during sampling. Note that $\bdv_{G_0}(c)=\widetilde{\bdv}_{G}(c)$.

By Corollary \ref{fct+} there is a function $F$ such that $F(b_1\cdot\abs{t_1},\dots,b_\ell\cdot\abs{t_\ell})$ is the largest multiple of $g$ which is not a conical combination of the integers $(b_1\cdot\abs{t_1}),\dots,(b_\ell\cdot\abs{c_l})$. 
	Let \[y':=\max\{y\mid n'-y\cdot\lcm(b_1,\dots,b_\ell)\cdot\abs{G_0}>F(b_1\cdot\abs{t_1},\dots,b_\ell\cdot\abs{t_\ell})\}.\]
We can see that $g$ divides $n'-y\cdot\lcm(b_1,\dots,b_\ell)\cdot\abs{G_0}$, as by assumption $g$ divides $n'$ and it is easy to see $g$ divides $\lcm(b_1,\dots,b_\ell)\cdot\abs{G_0}$.
	Thus it follows that there exist $p_i\in\mathbb N$ for $i\in[\ell]$ and a graph $\Gfill$ which consists of the disjoint union of 
	$p_i\cdot b_i$ components of type $t_i$ for  
	each $i\in[\ell]$, such that $\abs{\Gfill}=n'-y'\cdot\lcm(b_1,\dots,b_\ell)\cdot\abs{G_0}$.

Let $G'$ be the disjoint union of
\begin{itemize}
	\item $\ol{u}[s_i]$ components of type $s_i$ for each $i\in[m]$
	\item $k_i$ components of type $t_i$ for each $i\in[\ell]$
	\item $y'\cdot\lcm(b_1,\dots,b_\ell)$ copies of $G_0$
	\item $\Gfill$
\end{itemize}

We denote by $H$ the subgraph of $G'$ consisting of $y'\cdot\lcm(b_1,\dots,b_\ell)$ copies of $G_0$ and denote by $\Gfix$ the subgraph of $G'$ consisting of all connected components not in $H$ or $\Gfill$. We let $n_H\coloneq\abs{H}$, $\nfill\coloneq\abs{\Gfill}$, and $\nfix\coloneq\abs{\Gfix}$.
For $i\in[m]$ the number of components of $G'$ of type $s_i$ is equal to the $k$-CCHV entry $\ol{u}[s_i]$.
For $i\in[\ell]$ the number of components of $G'$ of type $t_i$ is both $\geq k$ and congruent $a_i(\Mod b_i)$ as there are $a_i(\Mod b_i)$ such components in $\Gfix$ and $0(\Mod b_i)$ such components in $H$ and in $\Gfill$.
Thus $G'\models\phi$. 

It remains to show that $G'$ is $\eps$-close to $G$.
First notice that $\bdv_{H}(c) = \bdv_{G_0}(c) = \widetilde{\bdv}_{G}(c)$ by construction. Since $q$ is sufficiently large we can apply lemma \ref{freqestlemma} to see that, with probability $\frac{19}{20}$, $\widetilde{\bdv}_{G}(c)$ is $\eps_2$-close to $\bdv_{G}(c)$, i.e. $\cab{\bdv_H(c) - \bdv_G(c)} \leq \eps_2$.

\medskip
To see that $\bdv_{H}(c)$ and $\bdv_{G'}(c)$ are $\eps_1$-close, for all $i\in\{1,\dots,M(c,d)\}$ let $h_i \coloneq \bhv_{H}(c)[i]$ and $f_i \coloneq \bhv_{\Gfill}(c)[i] + \bhv_{\Gfix}(c)[i]$, the numbers of vertices of $G'$ with ball type $\tau_i\in\Tball(c,d)$, in $H$ and in $\Gfix$ or $\Gfill$ respectively.

We know that $\sum_{i\in\{1,\dots,M(c,d)\}} h_i = n_H$ and $\sum_{i\in\{1,\dots,M(c,d)\}} f_i = \nfill + \nfix$. Thus the distribution vectors satisfy
\begin{align*}
	&\cab{\bdv_{G'}(c) - \bdv_{H}(c)} = \sum_{i\in\{1,\dots,M(c,d)\}} \abs{\bdv_{G'}(c)[i] - \bdv_{H}(c)[i]}  \\
	&= \sum_{i\in\{1,\dots,M(c,d)\}} \abs{\frac{h_i + f_i}{n} - \frac{h_i}{n_H}} 
	= \sum_{i\in\{1,\dots,M(c,d)\}} \abs{\frac{(n_H-n)h_i + n_H f_i}{nn_H}} \\
	&\leq  \sum_{i\in\{1,\dots,M(c,d)\}} \abs{\frac{(n_H-n)}{nn_H}h_i} +  \sum_{i\in\{1,\dots,M(c,d)\}} \abs{\frac{1}{n}f_i}  \\
	&= \frac{n-n_H}{n} + \frac{\nfix + \nfill}{n} 
	= \frac{2(\nfix + \nfill)}{n} < \eps/2 \hspace{1cm} \left(\tn{as }n>n_0> \frac{4(\nfix + \nfill)}{\eps}\right).
\end{align*}

\medskip
Finally $\cab{\bdv_{G'}(c) - \bdv_{G}(c)}\leq \cab{\bdv_{G'}(c) - \bdv_{H}(c)} + \cab{\bdv_{H}(c) - \bdv_{G}(c)} \leq \eps/2 + \eps/2 = \eps$ by the triangle inequality and thus $G'$ is $\eps$-close to $G$ by direct consequence of Lemma \ref{freqclosegraphclose}.
\end{proof}

We obtain the following extension to $\cmso$ as
an immediate consequence of Theorem~\ref{thm:testingFO+MOD}.

\begin{corollary}\label{cor:testingCMSO}
Let $c,d\in\N$. 
For every sentence $\phi$ of $\cmso$, the property  $P_\phi\cap\CC^c_d$
	 is uniformly testable in constant time on $\CC^c_d$. 
\end{corollary}
\begin{proof}
	See appendix \ref{app:testingCMSO}.
\end{proof}

\section{Conclusion}\label{sec:conclusion}
We showed that on graph classes containg all graphs that have both a fixed upper bound on the degree and on the component size, every
$\fomod$-property is testable with constant running time (Theorem~\ref{thm:testingFO+MOD}).
Since $\cmso$ and $\fomod$ have the same
expressive power on these classes, it follows that every $\cmso$-property is testable on these graph classes as well (Corollary~\ref{cor:testingCMSO}).
Indeed, it is not hard to see that this theorem can be generalised to relational structures of bounded degree, in the model introduced in~\cite{AdlerH18}, 
and we get the following.
\begin{theorem} Let $\sigma$ be a finite relational signature and let $d,c\in\mathbb N$.
	Let $\CC^c_d$ be the class consisting of all $\sigma$-structures that 
	have both degree at most $d$ and connected component size at most $c$. 
	For every sentence $\phi$ of $\fomod[\sigma]$ the property $P_\phi\cap\CC^c_d$ is testable with constant running time on $\mathcal C$.
\end{theorem}
We see this as a step towards answering an open question from~\cite{AdlerH18}, which asks if $\cmso$ is \emph{constant time} testable on classes that contain all relational
structures which have both degree at most $d$ and treewidth at most $t$. However, the question remains open.


\bibliographystyle{plainurl}
\bibliography{biblio}


\appendix
\section{Proof of lemma \ref{hnfuniformradius}}\label{app:hnfuniformradius}

Let $d\in\N$. For every sentence $\phi$ of $\fomod$ there exist an $\rmax\in\N$ and an $\fomod$ sentence $\psi$ such that $\phi \equiv_d \psi$, where $\psi$ is in HNF and every Hanf sentence has radius $\rmax$.

\begin{proof}
Applying \ref{hanflocalityfo} we may assume w.l.o.g. that $\phi$ is in HNF. 

Each Hanf sentence of $\psi$ is either $\exists^{\geq m}x(r,\tau)$ for some $r,m\in\N$ and $\tau\in\Tball(r,d)$ or $\exists^{j (\Mod \ell)}x(r,\tau)$ for some $r\in\N$, $\tau\in\Tball(r,d)$ and $j,\ell\in\N$ such that $j<\ell$. Let $\rmax$ be the largest radius $r$ of a Hanf sentence of $\phi$.

Recall that for every $r\leq\rmax$, $\Tball(r,d)\subseteq\Tball(\rmax,d)$. Let $S\subseteq\Tball(\rmax,d)$ be of maximum cardinality such that every 
$\tau'\in S$ is a supertype of $\tau$. We ennumerate $S=\{\tau'_1,\dots,\tau'_s\}$.

Let $X\coloneq[s]$, $Y_1\coloneq[m-1]\cup\{\quo{\geq m}\}$ and $Y_2\coloneq\{0\}\cup[\ell -1]$. We define two sets of $s$-tuples

\[B_1=\left\{ \ol{b}\in Y_1^X \mid \sum_{i=1}^s \ol{b}[i] \geq m \textnormal{ or } \ol{b}[i]=\quo{\geq m} \textnormal{ for some } i\in[s] \right\}\]

\[B_2=\left\{ \ol{b}\in Y_2^X \mid \sum_{i=1}^s \ol{b}[i] \equiv j (\Mod \ell) \right\}.\]

We replace $\exists^{\geq m}x(r,\tau)$ with the sentence

\[\biglor_{\ol{b}\in B_1}\left( \bigland_{i=1}^s \chi^{\ol{b}}_i \right)\]

where \[\chi^{\ol{b}}_i \coloneq
\begin{cases}
	\exists^{=\ol{b}[i]}x(\rmax,\tau'_i) & \textnormal{if } \ol{b}[i] \in [m-1] \\
	\exists^{\geq m}x(\rmax,\tau'_i) & \textnormal{if } \ol{b}[i] = \quo{\geq m}
\end{cases} \hspace{1cm} \tn{ for } i\in[s].
\]

We replace $\exists^{j (\Mod \ell)}x(r,\tau)$ with the sentence

\[\biglor_{\ol{b}\in B_2}\left( \bigland_{i=1}^s \left( \exists^{\ol{b}[i] (\Mod \ell)}x(\rmax,\tau'_i) \right) \right).\]

Let $\psi$ be the sentence obtained by performing the above replacement for each Hanf sentence of $\phi$ with radius less than $\rmax$. It is easy to see that $\psi \equiv_d \phi$ and the radius of every Hanf sentence of $\psi$ is $\rmax$.
\end{proof}

\section{Proof of lemma \ref{hnfuniformcap}}\label{app:hnfuniformcap}

Let $d\in\N$. For every sentence $\phi$ of $\fomod$, there exist a $k\in\N$ and an $\fomod$ sentence $\psi$ such that $\phi \equiv_d \psi$, where
$\psi$ is in disjunctive normal form (DNF) of atoms of one of the forms
\begin{itemize}
	\item $\exists^{\geq k}x(r,\tau)$ for some $r\in\N$ and $\tau\in\Tball(r,d)$
	\item $\exists^{=m}x(r,\tau)$  for some $m\in\N$ s.t. $m<k$, $r\in\N$ and $\tau\in\Tball(r,d)$
	\item $\exists^{j(\Mod \ell)}x(r,\tau)$  for some $j,\ell\in\N$ s.t. $j<\ell$, $r\in\N$ and $\tau\in\Tball(r,d)$.
\end{itemize}
In particular these literals may only have radii and types from among those of Hanf sentences of $\phi$.

\begin{proof}
Applying \ref{hanflocalityfo} we may assume w.l.o.g.\ that $\phi$ is in HNF. We may further assume that in particular $\phi$ is a disjunctive normal form with (possibly negated) literals from among Hanf sentences $\exists^{\geq m}x(r,\tau)$  for some $m\in\N$, $r\in\N$ and $\tau\in\Tball(r,d)$ and $\exists^{y(\Mod z)}x(r,\tau)$  for some $y,z\in\N$ s.t. $y<z$, $r\in\N$ and $\tau\in\Tball(r,d)$.
Let \[k \coloneq \max\{m\mid \exists^{\geq m}x(r,\tau) \textnormal{ or } \neg(\exists^{\geq m}x(r,\tau)) \textnormal{ is a literal of } \phi \textnormal{ for some } r,\tau \}.\]

For each conjunctive clause of the form $(\exists^{\geq m}x(r,\tau)\land\mu)$ of $\phi$ such that $m<k$, we replace the clause with

\[\left( \biglor_{m \leq i < k} \left(\exists^{=i}x(r,\tau) \land \mu\right) \right)\lor (\exists^{\geq k}x(r,\tau) \land\mu).\]

For each conjunctive clause of the form $(\neg\exists^{\geq m}x(r,\tau) \land \mu)$ of $\phi$ such that $m<k$, we replace the clause with

\[\left( \biglor_{1 \leq i < m} \left(\exists^{=i}x(r,\tau) \land \mu\right) \right).\]

For each conjunctive clause of the form $(\neg\exists^{j(\Mod \ell)}x(r,\tau)\land\mu)$ of $\phi$, we replace the clause with

\[\left( \biglor_{i\in\{0,\dots,z-1\}\setminus\{j\}}\left(\exists^{i(\Mod \ell)}x(r,\tau)\land\mu \right)\right).\]

Let $\psi$ be the sentence obtained after performing maximally many such replacements. We can see that this takes finitely many steps and that $\psi\equiv_d\phi$ as required. Every sentence added during a replacement must have a radius and type already appearing in $\phi$.
\end{proof}

\section{Proof of corollary \ref{fohistsentence}}\label{app:fohistsentence}

Let $c,d\in\N$, $c\geq1$. For every sentence $\phi$ of $\fomod$ there exist an $\fomod$ sentence $\psi$ and $k\in\N$ such that $\phi \equiv^c_d \psi$ and $\psi$ is a DNF of atoms of one of the forms
\begin{itemize}
	\item $\exists^{\geq k\cdot\rep(\tau)}x(c-1,\tau)$ for some $\tau\in\Tball(c-1,d)$
	\item $\exists^{=m}x(c-1,\tau)$  for some $m\in\N$ s.t. $m<k\cdot\rep(\tau)$ and $\tau\in\Tball(c-1,d)$
	\item $\exists^{j(\Mod \ell)}x(c-1,\tau)$  for some $j,\ell\in\N$ s.t. $j<\ell$ and $\tau\in\Tball(c-1,d)$.
\end{itemize}

\begin{proof}
We may assume w.l.o.g. that $\phi$ is in Hanf normal form. First transform $\phi$ as in the proof for Lemma \ref{hnfuniformradius} with the additional condition that $\rmax \geq c$. Now we apply Lemma \ref{hnfuniformcap} to get $\psi \equiv_d \phi$ such that $\psi$ is a DNF of atomic sentences of one of the forms
\begin{itemize}
	\item $\exists^{\geq k}x(\rmax,\tau)$ for some $\tau\in\Tball(\rmax,d)$
	\item $\exists^{=m}x(\rmax,\tau)$  for some $m\in\N$ s.t. $m<k$ and $\tau\in\Tball(\rmax,d)$
	\item $\exists^{j(\Mod \ell)}x(\rmax,\tau)$  for some $j,\ell\in\N$ s.t. $j<\ell$ and $\tau\in\Tball(\rmax,d)$.
\end{itemize}
Let $\tau^*\in\Tball(c-1,d)$ be the type of a singleton vertex. Define $\fomod$ sentences $\T\coloneq \left(\exists^{0(\Mod 1)}x(c-1,\tau^*)\right)$ and $\F\coloneq \left(\exists^{0(\Mod 2)}x(c-1,\tau^*)\right) \land \left(\exists^{1(\Mod 2)}x(c-1,\tau^*)\right)$.

Replace each atom of $\psi$ as follows:
\[\exists^{\geq k}x(\rmax,\tau) \mapsto\begin{cases}
\exists^{\geq k}x(c-1,\tau) & \tn{if } \tau\in\Tball(c-1,d) \\
\F & \tn{if } \tau\notin\Tball(c-1,d)
\end{cases}\]

\[\exists^{=m}x(\rmax,\tau) \mapsto \begin{cases}
\exists^{=m}x(c-1,\tau) & \tn{if } \tau\in\Tball(c-1,d) \\
\T & \tn{if } \tau\notin\Tball(c-1,d) \tn{ and } m=0 \\
\F & \tn{if } \tau\notin\Tball(c-1,d) \tn{ and } m\neq0
\end{cases}\]

\[\exists^{j(\Mod \ell)}x(\rmax,\tau) \mapsto\begin{cases}
\exists^{j(\Mod \ell)}x(c-1,\tau) & \tn{if } \tau\in\Tball(c-1,d) \\
\T & \tn{if } \tau\notin\Tball(c-1,d) \tn{ and } j=0 \\
\F & \tn{if } \tau\notin\Tball(c-1,d) \tn{ and } j\neq0
\end{cases}\]

Let $\psi'$ be the sentence obtained from $\psi$ by performing the replacements above. We can easily see $\psi'\equiv^c_d\psi$ as for every vertex $v$ of a $c$-finitary graph $G$ we have that $B_G(v,\rmax)=B_G(v,c-1)$ as $\rmax\geq c$ and the maximum radius of a connected component with $c$ vertices is $c-1$.

For each conjunctive clause of the form $(\exists^{\geq k}x(c-1,\tau) \land \mu)$ of $\psi'$ such that $\rep(\tau)>1$, we replace the clause with

\[\left( \biglor_{i=k}^{k\cdot\rep(\tau)-1} \left(\exists^{=i}x(c-1,\tau) \land \mu\right) \right)\lor (\exists^{\geq k\cdot\rep(\tau)}x(c-1,\tau) \land\mu).\]

Let $\psi''$ be the sentence obtained from $\psi'$ by making a finite, maximal number of such replacements. We can see that $\psi''\equiv^c_d\phi$ and $\psi''$ is a DNF of the required atoms.
\end{proof}

\section{Proof of lemma \ref{fohists}}\label{app:fohists}

Let $c,d\in\N$, $c\geq1$. For every satisfiable $\fomod$ sentence $\phi$ there exist a $k\in\N$ and a finite 
set $X$ of $k$-capped component histogram vectors such that $P_\phi \cap \CC^c_d$ contains exactly the 
$c$-finitary, $d$-bounded degree graphs which satisfy at least one of the histogram vectors in $X$, i.e. \[P_\phi \cap \CC^c_d = \{G\in \CC^c_d \mid G\textnormal{ satisfies $\ol{a}$ for some $\ol{a}\in X$}\}.\]

\begin{proof}
If $\phi$ is not satisfiable on $\CC^c_d$, then let $X=\es$. For the rest of the proof assume that $\phi$ is satisfiable on $\CC^c_d$.
By Corollary \ref{fohistsentence} we know that there is a sentence $\psi$ and $k\in\N$ such that $\phi \equiv^c_d \psi$ and $\psi$ is a DNF of atoms of the forms
\begin{itemize}
	\item $\exists^{\geq k\cdot\rep(\tau)}x(c-1,\tau)$ for some $\tau\in\Tball(c-1,d)$
	(a ``$\geq$-sentence")
	\item $\exists^{=m}x(c-1,\tau)$  for some $m\in\N$ s.t. $m<k\cdot\rep(\tau)$ and $\tau\in\Tball(c-1,d)$
	(an ``$=$-sentence")
	\item $\exists^{j(\Mod \ell)}x(c-1,\tau)$  for some $j,\ell\in\N$ s.t. $j<\ell$ and $\tau\in\Tball(c-1,d)$
	(a ``$\mod$-sentence").
\end{itemize}

For each conjunctive clause $\sigma=\chi_1\land\dots,\land\chi_s$ of $\phi$ such that $\sigma$ is satisfiable on $\CC^c_d$, we construct a finite set of $k$-CCHVs. The type of sentence $\chi_i$ is denoted $\tau_i$ for $i\in[s]$.
As $\sigma$ is satisfiable on $\CC^c_d$ we may assume $\underline{\tau_i}\in\Tcomp(c,d)$ for all $i\in[s]$.

For each $t\in\Tcomp(c,d)$, let $S^\sigma_t=\{\chi\in\{\chi_1,\dots,\chi_s\} \mid \underline{\tau_i}=t, i\in[s]\}$ be the set of sentences in the conjunction whose types have underlying component type 
$t$. For each $t\in\Tcomp(c,d)$ we define the set $e^\sigma_t$ by cases. It contains the possible entries of position $t$ of a $k$-CCHV vector ``describing'' the conjunctive clause.

\tbf{Case 1.} $S^\sigma_t$ is empty.
Let $e^\sigma_t=\{0\}\cup[k-1]\cup\{(0,1)\}$.

\tbf{Case 2.} $S^\sigma_t$ contains only ``$=$-sentence''s.
As $\sigma$ is satisfiable, for some $p\in\N$ each ``$=$-sentence'' $\exists^{=m}x(c-1,\tau)\in S^\sigma_t$ is such that $\frac{m}{\rep\tau}=p$. Otherwise the sentences assert that components of type $t$ appear a contradictory number of times.
Let $e^\sigma_t=\{p\}$.

\tbf{Case 3.} $S^\sigma_t$ contains only ``$\geq$-sentence''s.
Let $e^\sigma_t=\{(0,1)\}$.

\tbf{Case 4.} $S^\sigma_t$ contains only ``$\mod$-sentence''s. Each ``$\mod$-sentence'' $\exists^{j(\Mod \ell)}x(c-1,\tau)$ requires the number $p$ of $\tau$-type vertices to satisfy $p\equiv j(\mod\ell)$ and $p\equiv 0(\mod\rep(\tau))$. As $\sigma$ is satisfiable this has a solution $p\equiv i (\mod \lcm(\ell,\rep(\tau)))$ \cite{Jonesjones98}.
Let $z$ be the number of type $t=\underline{\tau}$ components. It follows that $z$ satisfies $z\equiv \frac{i}{\rep(\tau)} (\Mod \frac{\lcm(\ell,\rep(\tau))}{\rep(\tau)})$ as $z\cdot\rep(tau)=p$.
The collection $S^\sigma_t$ of ``$\mod$-sentence''s therefore defines a finite system of modular equations for $z$. As $\sigma$ is satisfiable, this system has a unique solution $z\equiv r(\mod n)$ for some $n\in\N$ depending only on the equations \cite{Jonesjones98}.
Let $e^\sigma_t=\{(r,n)\}\cup\{z\in[k-1] \mid z\equiv r(\mod n)\}$.

\tbf{Case 5.} $S^\sigma_t$ contains ``$\mod$-sentence''s and ``$\geq$-sentence''s.
As in case 4,  the ``$\mod$-sentence''s define a finite system of modular equations for the number $z$ of components of type $t$. This has solution $z\equiv r(\mod n)$.
Let $e^\sigma_t=\{(r,n)\}$.

\tbf{Case 6.} $S^\sigma_t$ contains ``$\mod$-sentence''s and ``$=$-sentence''s.
As in case 4,  the ``$\mod$-sentence"s define a finite system of modular equations for the number $z$ of components of type $t$. This has solution $z\equiv r(\mod n)$. As $\sigma$ is satisfiable, for some $p\in\N$ each ``$=$-sentence'' $\exists^{=m}x(c-1,\tau)$ is such that $\frac{m}{\rep\tau}=p$ and further $p\equiv r(\mod n)$.
Let $e^\sigma_t=\{p\}$.

These cases are exhaustive as $\phi$ is satsifiable and therefore ``$=$-sentence''s and ``$\geq$-sentence''s cannot co-occur in $S^\sigma_t$.

Let $X_\sigma=\{\ol{a} \mid \ol{a}[t]\in e^\sigma_t \tn{ for } t\in\Tball(c,d)\}$ and let $X=\bigcup_{\sigma \tn{ conjunctive clause of  }\phi}(X_\sigma)$. It is easy to verify that $X$ satisfies the statement of the lemma by construction.
\end{proof}

\section{Proof of corollary \ref{cor:testingCMSO}}\label{app:testingCMSO}

Let $c,d\in\N$. 
For every sentence $\phi$ of $\cmso$, the property  $P_\phi\cap\CC^c_d$ 
is uniformly testable in constant time on $\CC^c_d$.

We will make use of the following theorem of Eickmeyer, Elberfeld and Harwath \cite{EickmeyerEH17}.

\begin{theorem}[\cite{EickmeyerEH17}]
		Order-invariant $\mso$ and $\fomod$ have the same expressive power on graphs of bounded tree-depth.
	\end{theorem}

	As it is also known that order-invariant $\mso$ has at least the expressive power of $\cmso$ on any graph class~\cite{EickmeyerEH17}. 
It follows that on graphs of bounded tree-depth, $\fomod$ has at least the expressive power of $\cmso$.
Further each sentence of $\fomod$ is also a sentence of $\cmso$ as $\cmso$ extends $\fomod$. Thus on classes of bounded tree depth $\cmso$ and $\fomod$ have the same expressive power.

Corollary \ref{cor:testingCMSO} is now an immediate consequence of Theorem \ref{thm:testingFO+MOD} and the fact that finitary graph classes have bounded tree-depth.

\end{document}

%% file: macros.tex
\usepackage{pgfplots}
\usetikzlibrary{calc}
\usepackage{amsmath}
\usepackage{algorithm2e}
\RestyleAlgo{ruled}
\usepackage{array}
\usepackage{float}

\usetikzlibrary{decorations.pathmorphing, arrows.meta}

\tikzset{decoration={snake,amplitude=.4mm,segment length=2mm, post length=2mm, pre length=0mm},}
\definecolor{ao}{rgb}{0.0, 0.5, 0.0}

\newcommand{\tn}{\textnormal}

\newcommand{\ie}{\text{i.\,e.}\ }
\newcommand{\eg}{\emph{e.\,g.}\ }

\newcommand{\N}{\mathbb N}

\renewcommand{\P}{\mathcal{P}}
\newcommand{\Z}{\mathbb Z}
\newcommand{\CC}{\mathcal{C}}
\newcommand{\Zpos}{\Z^+}
\newcommand{\coloneq}{:=}
\newcommand{\eps}{\varepsilon}
\renewcommand{\subset}{\subsetneq}
\newcommand{\bigland}{\bigwedge}
\newcommand{\biglor}{\bigvee}
\newcommand{\ol}{\overline}
\newcommand{\rmax}{r_\textnormal{max}}
\newcommand{\T}{\textsc{True}}
\newcommand{\F}{\textsc{False}}
\newcommand{\es}{\emptyset}

\newcommand{\dist}{\operatorname{dist}}

\newcommand{\B}{\operatorname{B}}
\newcommand{\ball}{\operatorname{ball}}

\renewcommand{\phi}{\varphi}
\renewcommand{\epsilon}{\varepsilon}

\newcommand{\lcm}{\operatorname{lcm}}
\newcommand{\rep}{\operatorname{rep}}
\newcommand{\Gfill}{G_{\tn{fill}}}
\newcommand{\Gfix}{G_{\tn{fix}}}
\newcommand{\nfill}{n_{\tn{fill}}}
\newcommand{\nfix}{n_{\tn{fix}}}

\DeclareMathSymbol{\mlq}{\mathord}{operators}{``}
\DeclareMathSymbol{\mrq}{\mathord}{operators}{`'}
\newcommand{\quo}[1]{{\mlq\mlq #1 \mrq\mrq}} 

\newcommand{\Tball}{T_{\textnormal{ball}}}
\newcommand{\Tcomp}{T_{\textnormal{comp}}}

\newcommand{\bdv}{\operatorname{bdv}}

\newcommand{\bhv}{\operatorname{bhv}}
\newcommand{\chv}{\operatorname{chv}}

\newcommand{\tbf}[1]{\textbf{#1}}

\providecommand{\abs}[1]{\left\lvert#1\right\rvert}
\providecommand{\norm}[1]{\left\lVert#1\right\rVert}
\newcommand{\cab}[1]{\norm{#1}_1}
\newcommand{\ceil}[1]{\left\lceil #1 \right\rceil}

\newcommand{\mso}{\operatorname{MSO}}
\newcommand{\cmso}{\operatorname{CMSO}}
\newcommand{\fo}{\textsc{FO}}
\newcommand{\fomod}{\operatorname{FOMOD}} 
\newcommand{\Mod}{\operatorname{mod}}
\renewcommand{\mod}{\Mod} 

